\documentclass[12pt]{article}

\usepackage[utf8]{inputenc}

\usepackage[style=alphabetic,
  sorting = nyt,
  sortcites=true,
  giveninits=true,
  date=year,
  isbn=false,
  maxbibnames=99,
  maxalphanames=6, 
  backend=biber]{biblatex}
\DeclareFieldFormat*{titlecase}{\MakeSentenceCase{#1}}

\DeclareSourcemap{\maps[datatype=bibtex]{\map{\step[fieldset=shorthand, null]}}}

\DeclareSourcemap{
  \maps[datatype=bibtex]{
    \map[overwrite]{
      \step[fieldsource=doi, final]
      \step[fieldset=url, null]
      \step[fieldset=eprint, null]
    }  
  }
}

\DeclareSourcemap{
  \maps[datatype=bibtex, overwrite]{
    \map{
      \step[fieldset=language, null]
      \step[fieldset=month, null]
      \step[fieldset=urldate, null]
    }
  }
}

\AtEveryBibitem{%
  \clearlist{language}%
  \clearlist{month}%
  \clearlist{urldate}%
  \clearfield{urldate}%
}

\renewbibmacro{in:}{}

\DeclareFieldFormat[misc]{title}{\mkbibquote{#1}}
\DeclareFieldFormat[article,inproceedings]{volume}{\mkbibbold{#1}}
\DeclareFieldFormat[inproceedings]{series}{\mkbibitalic{#1}}

\usepackage[english]{babel}
 \usepackage{amsmath}
\usepackage{amsthm}
\usepackage{amssymb}
\usepackage[breaklinks=true]{hyperref}
\usepackage{braket}
\usepackage{color}
\usepackage{comment}
\usepackage{graphicx}
\usepackage[margin=2.5cm]{geometry}
\usepackage[noblocks]{authblk}
\usepackage[noabbrev]{cleveref}

\usepackage{enumerate}

\usepackage{chngcntr}
\counterwithin{figure}{section}
\counterwithin{equation}{section}

\usepackage[UKenglish]{isodate}
\cleanlookdateon

\let\C\relax

\newcommand{\eps}{\varepsilon}
\newcommand{\E}{{\mathrm{e}}}
\newcommand{\I}{\mathrm{i}}
 \newcommand{\R}{ \mathbb{R} }
  \newcommand{\Sph}{ \mathbb{S} }
\newcommand{\C}{ \mathbb{C} }

\newcommand{\D}{\mathrm{d}}

 \newcommand{\norm}[1]{\left\Vert #1 \right\Vert} 
 \newcommand{\abs}[1]{\left\vert #1 \right\vert}
 
\newcommand{\longip}[3]{\left\langle #1 \middle\vert #2 \middle\vert #3 \right\rangle}

\DeclareMathOperator{\sgn}{sgn}
\DeclareMathOperator{\spec}{spec}

\DeclareMathOperator{\arcoth}{arcoth}

\newcommand{\ud}{\,\textnormal{d}}

\usepackage{seqsplit}

\theoremstyle{plain}
\newtheorem{thm}{Theorem}[section]

\newtheorem{prop}[thm]{Proposition}

\newtheorem{lemma}[thm]{Lemma}

\theoremstyle{definition}
\newtheorem{defi}[thm]{Definition}
\newtheorem{rmk}[thm]{Remark}
\newtheorem{assumption}[thm]{Assumption}

\newtheorem{remark}[thm]{Remark}

\crefname{thm}{theorem}{theorems}
\crefname{problem}{problem}{problems}
\crefname{lemma}{lemma}{lemmas}
\crefname{lem}{lemma}{lemmas}
\crefname{cor}{corollary}{corollaries}
\crefname{prop}{proposition}{propositions}
\crefname{conj}{conjecture}{conjectures}
\crefname{defn}{definition}{definitions}
\crefname{defi}{definition}{definitions}
\crefname{note}{note}{notes}
\crefname{ex}{example}{examples}
\crefname{remark}{remark}{remarks}
\crefname{rmk}{remark}{remarks}
\crefname{notation}{notation}{notations}
\crefname{assumption}{assumption}{assumptions}
\crefname{claim}{claim}{claims}
\crefname{claim*}{claim}{claims}

\allowdisplaybreaks

\title{Universality in low-dimensional BCS theory}
\author{\small Joscha Henheik\footnote{\href{mailto:joscha.henheik@ist.ac.at}{joscha.henheik@ist.ac.at}} \hspace{1.2cm} Asbjørn Bækgaard Lauritsen\footnote{\href{mailto:alaurits@ist.ac.at}{alaurits@ist.ac.at}} \hspace{1.2cm} Barbara Roos\footnote{\href{mailto:barbara.roos@ist.ac.at}{barbara.roos@ist.ac.at}} \\ Institute of Science and Technology Austria, Am Campus 1, 3400 Klosterneuburg, Austria}

\begin{document}
\maketitle
\begin{abstract}
It is a remarkable property of BCS theory that the ratio of the energy gap at zero temperature $\Xi$ and the critical temperature $T_c$ is (approximately) given by a universal constant, independent of the microscopic details of the fermionic interaction. This universality has rigorously been proven quite recently in three spatial dimensions and three different limiting regimes: weak coupling, low density, and high density. The goal of this short note is to extend the universal behavior to lower dimensions $d=1,2$ and give an exemplary proof in the weak coupling limit. 
\end{abstract}

\section{Introduction}
The Bardeen–Cooper–Schrieffer (BCS) theory of superconductivity \cite{bcs.original} is governed by the \emph{BCS gap equation}. 
For translation invariant systems without external fields the BCS gap equation is 
\begin{equation}\label{eqn.bcs.gap.eqn}
  \Delta(p) = - \frac{1}{(2\pi)^{d/2}} \int_{\R^d} \hat{V}(p-q) \frac{\Delta(q)}{E_\Delta(p)}\tanh \left(\frac{E_\Delta(p)}{2T}\right) \ud q
\end{equation}
with dispersion relation $E_\Delta(p) = \sqrt{(p^2- \mu)^2 + |\Delta(p)|^2}$. Here, $T \ge 0$ denotes the temperature
and $\mu>0$ the chemical potential.
We consider dimensions
$d \in \set{1,2,3}$.
The Fourier transform of the potential $V\in L^1(\R^d) \cap L^{p_V}(\R^d)$ (with a $d$-dependent $p_V \ge 1$ to be specified below), modeling their effective interaction, is denoted by $\hat{V}(p) = (2 \pi)^{-d/2} \int_{\R^d} V(x) \E^{-\I p \cdot x} \D x$.

According to BCS theory, a system is in a superconducting state, if there exists a non-zero solution $\Delta$ 
to the gap equation \eqref{eqn.bcs.gap.eqn}. 
The question of existence of such a non-trivial solution $\Delta$ hinges, in particular, on the temperature $T$.
It turns out, there exists a critical temperature $T_c\geq 0$ such that for $T < T_c$ there exists a non-trivial solution, 
and for $T\geq T_c$ it does not \cite[Theorem \ref{thm:lincrit} and Definition \ref{def:Tc}]{hainzl.hamza.seiringer.solovej}. 
This critical temperature is one of the key (physically measurable) quantities of the theory and its asymptotic behavior, in three spatial dimensions, has been studied in three physically rather different limiting regimes: 
In a weak-coupling limit (i.e.~replacing $V \to \lambda V$ and taking $\lambda \to 0$) \cite{Hainzl.Seiringer.2008,frank.hainzl.naboko.seiringer}, in a low-density limit (i.e.~$\mu \to 0$) \cite{hainzl.seiringer.scat.length}, 
and in a high-density limit (i.e.~$\mu \to \infty$) \cite{Henheik.2022}.

As already indicated above, at zero temperature, the function $E_\Delta$ may be interpreted as the dispersion relation of a certain `approximate' Hamiltonian of the quantum many-body system, see \cite[Appendix A]{hainzl.hamza.seiringer.solovej}. 
In particular 
\begin{equation} \label{eq:energygap}
  \Xi := \inf_{p \in \R^d} E_\Delta(p)
\end{equation}
has the interpretation of an energy gap associated with the approximate BCS Hamiltonian and as such represents a second key quantity of the theory. Analogously to the critical temperature, 
the asymptotic behavior of this energy gap, again in three spatial dimensions, has been studied in the same three different limiting regimes:
In a weak coupling limit \cite{Hainzl.Seiringer.2008}, in a low density limit \cite{lauritsen.energy.gap.2021},
and in a high density limit \cite{Henheik.Lauritsen.2022}. 

In this paper, we focus on a remarkable feature of BCS theory, which is well known in the physics literature 
\cite{bcs.original,nozieres.schmitt-rink,langmann.et.al.2019}: 
The ratio of the energy gap $\Xi$ and critical temperature $T_c$ tends to a \emph{universal constant, independent of the microscopic details} of the interaction between the fermions, i.e.~the potential $V$. More precisely, in three spatial dimension, it holds that
\begin{equation} \label{eq:univintro}
  \frac{\Xi }{T_c} \approx \frac{\pi}{\E^{\gamma}}  \approx 1.76\,,
\end{equation}
where $\gamma\approx 0.577$ is the Euler-Mascheroni constant, in each of the three physically very different limits mentioned above. This result follows as a limiting equality by combining asymptotic formulas for the critical temperature $T_c$ (see \cite{frank.hainzl.naboko.seiringer, Hainzl.Seiringer.2008, hainzl.seiringer.scat.length, Henheik.2022}) and the energy gap $\Xi$ (see \cite{Hainzl.Seiringer.2008, lauritsen.energy.gap.2021, Henheik.Lauritsen.2022}) in the three different regimes.  Although these scenarios (weak coupling, low density, and high density) are physically rather different, they all have in common that `superconductivity is weak' and one can hence derive an asymptotic formula for $T_c$ and $\Xi$ as they depart from being zero (in the extreme cases $\lambda = 0$, $\mu = 0$, $\mu = \infty$, respectively). However, all the asymptotic expressions are \emph{not perturbative}, as they depend exponentially on the natural dimensionless small parameter in the respective limit. We refer to the above mentioned original works for details.

The goal of this note is to prove the \emph{same} universal behavior \eqref{eq:univintro}, which has already been established in three spatial dimension, also in dimensions $d=1,2$ in the weak coupling limit (i.e.~replacing $V \to \lambda V$ and taking $\lambda \to 0$). This situation serves as a showcase for the methods involved in the proofs of the various limits in three dimensions (see Remark~\ref{rmk:Tccomp} and Remark~\ref{rmk:Xicomp} below). 
Apart from the mathematical curiosity in $d=1,2$, there have been recent studies in lower-dimensional superconductors in the physics literature, out of which we mention one-dimensional superconducting nanowires \cite{natarajan2012superconducting} and two-dimensional `magic angle' graphene \cite{Cao.Fatemi.ea.2018}. 

In the remainder of this introduction, we briefly recall the mathematical formulation of BCS theory, which has been developed mostly by Hainzl and Seiringer, but also other co-authors \cite{hainzl.hamza.seiringer.solovej, frank.hainzl.naboko.seiringer, hainzl.seiringer.16}.
Apart from the universality discussed here, also many other properties of BCS theory have been shown using this formulation:
Most prominently, Ginzburg-Landau theory, as an effective theory describing superconductors close to the critical temperature, has been derived from BCS theory \cite{Frank.Hainzl.ea.2012,Frank.Lemm.2016,deuchert_microscopic_2022,Deuchert.Hainzl.ea.2022a}. More recently, it has been shown that the effect of boundary superconductivity occurs in the BCS model \cite{Hainzl.Roos.ea.2022}. 
We refer to \cite{hainzl.seiringer.16} for a more comprehensive review of developments in the mathematical formulation of BCS theory. 
The universal behavior in the weak coupling limit for lower dimensions $d=1,2$ is presented in Section \ref{sec:univ12}. Finally, in Section \ref{sec:proofs}, we provide the proofs of the statements from Section \ref{sec:univ12}.

\subsection{Mathematical formulation of BCS theory}
We will now briefly recall the mathematical formulation \cite{hainzl.hamza.seiringer.solovej, hainzl.seiringer.16} of BCS theory \cite{bcs.original}, which is an effective theory developed for describing superconductivity of a fermionic gas. In the following, we consider these fermions in $\R^d$, $d = 1,2$, at temperature $T \ge 0$ and chemical potential $\mu \in \R$, interacting via a two-body potential $V$, for which we assume the following. 
\begin{assumption} \label{ass:basic}
We have that $V$ is real-valued, reflection symmetric, i.e.~$V(x) = V(-x)$ for all $x \in \R^d$, and it satisfies $V \in L^{p_V}(\R^d)$, where $p_V = 1$ if $d = 1$, $p_V \in (1,\infty)$ if $d = 2$.
\end{assumption}
\noindent Moreover, we neglect external fields, in which case the system is translation invariant.

The central object in the mathematical formulation of the theory is the BCS functional, which can naturally be viewed as a function of BCS states $\Gamma$. These states are given by a pair of functions $(\gamma, \alpha)$ and can be conveniently represented as a $2 \times 2$ matrix valued Fourier multiplier on $L^2(\R^d) \oplus L^2(\R^d)$ of the form
\begin{equation} \label{eq:Gamma}
	\hat{\Gamma}(p) = \begin{pmatrix}
		\hat{\gamma}(p) & \hat{\alpha}(p)\\ \overline{\hat{\alpha}(p)} & 1- \hat{\gamma}(p)
	\end{pmatrix}
\end{equation}
for all $p \in \R^d$. In \eqref{eq:Gamma}, $\hat{\gamma}(p)$ denotes the Fourier transform of the one particle density matrix and $\hat{\alpha}(p)$ is the Fourier transform of the Cooper pair wave function. We require reflection symmetry of $\hat{\alpha}$, i.e.~$\hat{\alpha}(-p) = \hat{\alpha}(p)$, as well as $0 \le \hat{\Gamma}(p) \le 1$ as a matrix.

The \emph{BCS free energy functional} takes the form
\begin{equation} \label{eq:functional}
  \mathcal{F}_T[\Gamma]
    := \int_{\R^d}(p^2 - \mu) \hat{\gamma}(p) \D p - T S[\Gamma] + \int_{\R^d} V(x) |\alpha(x)|^2 \D x\,, 
    \qquad \Gamma \in \mathcal{D},
\end{equation}
\[
    \mathcal{D} 
    := \left\{   \hat{\Gamma}(p) = \begin{pmatrix}
    \hat{\gamma}(p) & \hat{\alpha}(p)\\ \overline{\hat{\alpha}(p)} & 1- \hat{\gamma}(p)
    \end{pmatrix}
      : 
      0 \le \hat{\Gamma} \le 1\,, \ \hat{\gamma} \in L^1(\R^d, (1 + p^2) \D p)\,, \ \alpha \in H^1_{\rm sym}(\R^d)  
      \right\}\,,
\]
where the \emph{entropy density} is defined as
\[
  S[\Gamma] = - \int_{\R^d} \mathrm{Tr}_{\C^2} \left[\hat{\Gamma}(p) \log \hat{\Gamma}(p)\right] \D p\,.
\]
The minimization problem associated with \eqref{eq:functional} is well defined. In fact, the following result has only been proven for $d=3$ and $V \in L^{3/2}(\R^3)$, but its extension to $d= 1, 2$ is straightforward. 
\begin{prop}[{\cite{hainzl.hamza.seiringer.solovej}, see also \cite{hainzl.seiringer.16}}] \label{prop:exofmin}
Under Assumption \ref{ass:basic} on $V$, the BCS free energy is bounded below on $\mathcal{D}$ and attains its minimum. 
\end{prop}

\noindent
The BCS gap equation \eqref{eqn.bcs.gap.eqn} arises as the Euler--Lagrange equations of this functional \cite{hainzl.hamza.seiringer.solovej}.
Namely by defining $\Delta = -2\widehat{V\alpha}$, the Euler--Lagrange equation for $\alpha$ takes the form of the BCS gap equation \eqref{eqn.bcs.gap.eqn}.
Additionally, one has the following linear criterion for the BCS gap equation to have non-trivial solutions. 
Again, so far, a proof has only been given in spatial dimension $d=3$ and for $V \in L^{3/2}(\R^3)$, but its extension to $d= 1, 2$ is straightforward. 
\begin{thm}[{\cite[Thm.~1]{hainzl.hamza.seiringer.solovej}}] \label{thm:lincrit}
Let $V$ satisfy Assumption \ref{ass:basic} and let $\mu\in\R$ as well as $T\geq 0$. 
Then, writing $\mathcal{F}_T[\Gamma] \equiv \mathcal{F}_T(\gamma, \alpha)$, the following are equivalent. 
\begin{enumerate}
\item The minimizer of $\mathcal{F}_T$ is not attained with $\alpha = 0$, i.e. 
\[
  \inf_{(\gamma, \alpha)\in \mathcal{D}} \mathcal{F}_T(\gamma, \alpha) < \inf_{(\gamma, 0) \in \mathcal{D}} \mathcal{F}_T(\gamma,0),
\]
\item 
There exists a pair $(\gamma, \alpha)\in \mathcal{D}$ with $\alpha \ne 0$ such that $\Delta = -2\widehat{V\alpha}$
satisfies the BCS gap equation \eqref{eqn.bcs.gap.eqn},
\item The linear operator $K_T + V$, where $K_T(p) = \frac{p^2-\mu}{\tanh((p^2-\mu)/(2T))}$ has at least one negative eigenvalue.
\end{enumerate}
\end{thm}

\noindent
The third item immediately leads to the following definition of the \emph{critical temperature} $T_c$ for the existence of non-trivial solutions of the BCS gap equation \eqref{eqn.bcs.gap.eqn}. 
\begin{defi}[{Critical temperature, see \cite[Def.~1]{frank.hainzl.naboko.seiringer}}]
\label{def:Tc}
For $V$ satisfying \Cref{ass:basic},
we define the critical temperature $T_c \ge 0$ as 
\begin{equation} \label{eq:Tc}
	T_c := \inf\{ T > 0 : K_T + V \ge 0 \}\,.
\end{equation}
By $K_T(p)\geq 2T$ and the asymptotic behavior $K_{T}(p) \sim p^2$ for $|p| \to \infty$, Sobolev's inequality \cite[Thm.~8.3]{analysis} implies that the critical temperature is well defined.
\end{defi}

\noindent
The other object we study is the energy gap $\Xi$ defined in \eqref{eq:energygap}.
The energy gap depends on the solution $\Delta$ of the gap equation \eqref{eqn.bcs.gap.eqn} at $T=0$.
A priori, $\Delta$ may not be unique.
However, for potentials with non-positive Fourier transform, this possibility can be ruled out. 
\begin{prop}[{see \cite[(21)-(22) and Lemma~2]{Hainzl.Seiringer.2008}}]
\label{prop:alphaunique}
Let $V$ satisfy Assumption \ref{ass:basic} (and additionally $V \in L^1(\R^2)$ in case that $d = 2$). Moreover, we assume that $\hat{V} \le 0$ and $\hat{V}(0)<0$. Then, there exists a unique minimizer $\Gamma$ of $\mathcal{F}_0$ (up to a constant phase in $\alpha$).
One can choose the phase such that $\alpha$ has strictly positive Fourier transform $\hat{\alpha} >0 $. 
\end{prop}

\noindent
In particular, we conclude that $\Delta$ is strictly positive. Moreover, by means of the gap equation \eqref{eqn.bcs.gap.eqn}, $\Delta$ is continuous and thus $\Xi > 0$.

\section{Main Results} \label{sec:univ12}
As explained in the introduction, our main result in this short note is the extension of the universality \eqref{eq:univintro} from $d=3$ to lower spatial dimensions $d=1,2$ in the limit of weak coupling (i.e., replacing $V \to \lambda V$ and taking $\lambda \to 0$). We assume the following properties for the interaction potential~$V$.
\begin{assumption}\label{assumption.d=2} Let $d \in \set{1,2}$ and assume that $V$ satisfies Assumption \ref{ass:basic} as well as $\hat V \leq 0$, $\hat V(0) < 0$. Moreover, for $d=1$ we assume that $(1 + |\cdot|^{\eps})V\in L^1(\R^1)$ for some $\eps > 0$. Finally, in case that $d=2$, we suppose that $V\in L^1(\R^2)$ is radial.
\end{assumption}

\noindent
By Proposition \ref{prop:alphaunique}, this means that, in particular, the minimizer of $\mathcal{F}_0$ is unique (up to a phase) and the associated energy gap at zero temperature \eqref{eq:energygap} is strictly positive, $\Xi > 0$. We are now ready to state our main result. 
\begin{thm}[BCS Universality in one and two dimensions] \label{thm.univ.d=2} 
Let $V$ be as in \Cref{assumption.d=2}. Then the critical temperature $T_c(\lambda)$ (defined in \eqref{eq:Tc}) and the energy gap $\Xi(\lambda)$ (defined in \eqref{eq:energygap}) are strictly positive for all $\lambda >0$ and it holds that
\[
	\lim_{\lambda \to 0} \frac{\Xi(\lambda)}{T_c(\lambda)} = \frac{\pi}{\E^{\gamma}}\,, 
\]
where $\gamma \approx 0.577$ is the Euler-Mascheroni constant.
\end{thm}

\noindent
To prove the universality, we separately establish asymptotic formulas for $T_c$ (see Theorem~\ref{thm:asympt12d Tc}) and $\Xi$ (see Theorem \ref{thm:asympt12d Xi}), valid to second order, and compare them by taking their ratio. 
The asymptotic formula for $T_c$ is valid under weaker conditions on $V$ than Assumption~\ref{assumption.d=2}, because we do not need uniqueness of $\Delta$.
To obtain the asymptotic formulas, we first introduce two self-adjoint operators $\mathcal{V}_\mu^{(d)}$ and $\mathcal{W}_\mu^{(d)}$ mapping $L^{2}(\Sph^{d-1}) \to L^{2}(\Sph^{d-1})$ and as such measuring the strength of the interaction $\hat{V}$ on the (rescaled) Fermi surface (see \cite{Hainzl.Seiringer.2008, Henheik.2022, Henheik.Lauritsen.2022}). 
To assure that $\mathcal{V}_\mu^{(d)}$ and $\mathcal{W}_\mu^{(d)}$ will be well-defined and compact, we assume the following.
\begin{assumption}\label{assumption.W} Let $V$ satisfy Assumption \ref{ass:basic}. Additionally, assume that for $d=1$, $\big(1 + (\ln(1 + |\cdot|))^2\big) V \in L^1(\R^1)$ and for $d=2$, $V \in L^1(\R^2)$.
\end{assumption}

\noindent
First, in order to capture the strength to leading order, we define $\mathcal{V}_\mu^{(d)}$ via
\[
  (\mathcal{V}_\mu^{(d)} u)(p) = \frac{1}{(2\pi)^{d/2}} \int_{\Sph^{d-1}} \hat V(\sqrt{\mu}(p-q)) u(q) \ud \omega(q)\,, 
\]
where $\ud\omega$ is the Lebesgue measure on $\Sph^{d-1}$. Since $V \in L^1(\R^d)$, we have that $\hat{V}$ is a bounded continuous function and hence $\mathcal{V}_\mu^{(d)}$ is a Hilbert-Schmidt operator (in fact, trace class with trace being equal to $(2 \pi)^{-d}\vert \Sph^{d-1} \vert  \int_{\R^d} V(x) \D x$). Therefore, its lowest eigenvalue $e_\mu^{(d)} := \inf \spec \mathcal{V}_\mu^{(d)}$ satisfies $e_\mu^{(d)} \le 0$ and it is strictly negative if e.g. $\int V <0$ as in Assumption~\ref{assumption.d=2}.

Second, in order to capture the strength of $\hat{V}$ to next to leading order, we define the operator $\mathcal{W}_\mu^{(d)}$ via its quadratic form 
\[
\begin{aligned}
  & \longip{u}{\mathcal{W}_\mu^{(d)}}{u}
  \\ & \quad = \mu^{d/2-1} 
    \left[\int_{|p| <\sqrt{2}} \frac{1}{|p^2 - 1|} \left(|\psi(\sqrt{\mu}p)|^2 - |\psi(\sqrt{\mu}p/|p|)|^2\right) \ud p 
    + \int_{|p| > \sqrt{2}} \frac{1}{|p^2 - 1|} |\psi(\sqrt{\mu}p)|^2 \ud p
    \right],
\end{aligned}
\]
where $\psi(p) = \frac{1}{(2\pi)^{d/2}} \int_{\Sph^{d-1}} \hat V(p - \sqrt{\mu}q) u(q) \ud \omega(q)$ and $u \in L^2(\Sph^{d-1})$. 
The proof of the following proposition shall be given in Section \ref{sec:pf_Wmu}.
\begin{prop} \label{prop:Wmu}
Let $d \in \set{1,2}$ and let $V$ satisfy Assumption~\ref{assumption.W}.
The operator $\mathcal{W}_\mu^{(d)}$ is well-defined and Hilbert-Schmidt. 
\end{prop}

\noindent
Next, we define the self-adjoint Hilbert-Schmidt operator 
$$\mathcal{B}_\mu^{(d)}(\lambda) :=\frac{\pi }{2} \left(\lambda \mathcal{V}_\mu^{(d)} - \lambda^2 \mathcal{W}_\mu^{(d)}\right)$$
 on $L^2(\Sph^{d-1})$ and its ground state energy
\begin{equation} \label{eq:bmudef}
b_\mu^{(d)}(\lambda) := \inf\spec\left( \mathcal{B}_\mu^{(d)} (\lambda)   \right).
\end{equation}
Note that if $e_\mu^{(d)}< 0$, then also $b_\mu^{(d)}(\lambda) <0$ for small enough $\lambda$. After these preparatory definitions, we are ready to state the separate asymptotic formulas for the critical temperature and the energy gap in one and two dimensions, which immediately imply Theorem \ref{thm.univ.d=2}.

\begin{thm}[Critical Temperature for $d=1,2$] \label{thm:asympt12d Tc} 
Let $\mu > 0$. Let $V$ satisfy Assumption~\ref{assumption.W} and additionally $e_\mu^{(d)}< 0 $. Then the critical temperature $T_c$, given in Definition \ref{def:Tc}, is strictly positive and satisfies
\begin{equation*}
\lim\limits_{\lambda \to 0} \left(  \ln\left(\frac{\mu}{T_c(\lambda)}\right) + \frac{\pi }{2 \, \mu^{d/2 - 1} \, b_\mu^{(d)}(\lambda)}\right) = - \gamma - \ln \left( \frac{2 c_d}{\pi } \right)\,, 
\end{equation*}
where $\gamma$ denotes the Euler-Mascheroni constant and $c_1=\frac{4}{1+\sqrt{2}}$ and $c_2=1$.
\end{thm}

\noindent
Here, the Assumptions on $V$ are weaker than Assumption~\ref{assumption.d=2}, since $\hat V(0)<0$ implies that $e_\mu^{(d)}<0$.
We thus have the asymptotic behavior
\begin{equation*}
T_c(\lambda) = 2 c_d \, \frac{\E^{\gamma}}{\pi} \, \big(1 + o(1)\big) \mu \, \E^{\pi/(2 \mu^{d/2- 1} b_\mu^{(d)}(\lambda))}
\end{equation*}
in the limit of small $\lambda$.

\begin{remark}\label{rmk.HS10}
\Cref{thm:asympt12d Tc} is essentially a special case of \cite[Theorem 2]{hainzl.seiringer.2010}. 
We give the proof here for two main reasons. 
\begin{itemize}
\item[(i)] There is still some work required to translate the statement of \cite[Theorem 2]{hainzl.seiringer.2010} into a form 
in which it is comparable to that of \Cref{thm:asympt12d Xi} (in order to prove \Cref{thm.univ.d=2}). 
The main difficulty is that the operator $\mathcal{W}_\mu^{(d)}$ in \cite{hainzl.seiringer.2010} is only defined via a limit, 
\cite[Equation (2.10)]{hainzl.seiringer.2010}.
\item[(ii)] The goal of this paper is to give an exemplary proof of \Cref{thm:asympt12d Tc} 
in order to compare it to the proofs of the similar statements in the literature
concerning the asymptotic behavior of the critical temperature in various limits \cite{hainzl.seiringer.scat.length,Hainzl.Seiringer.2008,Henheik.2022}.
\end{itemize}
\end{remark}

\Cref{thm:asympt12d Tc} is complemented by the following asymptotics for the energy gap. 
\begin{thm}[Energy Gap for $d=1,2$] \label{thm:asympt12d Xi} 
Let $V$ satisfy Assumption \ref{assumption.d=2} and let $\mu > 0$. Then there exists a unique radially symmetric minimizer (up to a constant phase) of the BCS functional \eqref{eq:functional} at temperature $T=0$. The associated energy gap $\Xi$, given in \eqref{eq:energygap}, is strictly positive and satisfies
\begin{equation*}
	\lim\limits_{\lambda \to 0} \left(  \ln\left(\frac{\mu}{\Xi}\right) + \frac{\pi }{2 \, \mu^{d/2 - 1} \, b_\mu^{(d)}(\lambda)}\right) = - \ln ( 2 c_d )\,, 
\end{equation*}
where $b_\mu^{(d)}$ is defined in \eqref{eq:bmudef} and $c_1=\frac{4}{1+\sqrt{2}}$ and $c_2=1$.
\end{thm}

\noindent
In other words, we have the asymptotic behavior
\begin{equation*}
\Xi(\lambda) = 2 c_d \,  \big(1 + o(1)\big) \mu \, \E^{\pi/(2 \mu^{d/2- 1} b_\mu^{(d)}(\lambda))}
\end{equation*}
in the limit of small $\lambda$.
Now, \Cref{thm.univ.d=2} follows immediately from \Cref{thm:asympt12d Tc,thm:asympt12d Xi}.

\begin{remark}[Other limits in dimensions $d=1,2$]
Similarly to the presented results, one could also consider the limits of low and high density, i.e.~$\mu\to 0$ and $\mu\to \infty$, respectively.
We expect that also here the universality $\frac{\Xi}{T_c} \approx \frac{\pi}{\E^{\gamma}}$ holds.
Indeed, one would expect that the proofs of BCS universality in dimension $d=3$ should carry over to one and two dimensions with some minor technical modifications. 
Note that, even for the (technically less demanding) case of a weak coupling limit, which we consider here, 
there are still some technical details that are different in dimensions $d=1,2$ compared to dimension $d=3$.
Hence, it is not a trivial matter to generalize the arguments of 
\cite{hainzl.seiringer.scat.length,lauritsen.energy.gap.2021,Henheik.2022,Henheik.Lauritsen.2022} to one and two dimensions.
Moreover, for the case of low density, there is even an issue of what exactly low density \emph{means} in dimensions one and two: In three spatial dimensions \cite{hainzl.seiringer.scat.length, lauritsen.energy.gap.2021}, the asymptotic formulas for $T_c$ and $\Xi$ were obtained for potentials $V$ with \emph{negative} scattering length but not creating bound states for the Laplacian. This latter condition ensures that $\mu \to 0$ actually corresponds to the limit of low density. However, in spatial dimensions one and two, attractive potentials, no matter how weak, always give rise to bound states of $-\nabla^2 + V$, see \cite{simon.1976}. Thus for $\mu = 0$ the particle density is non-zero. We will not deal with the low- and high-density limits here.
\end{remark}

\noindent
The rest of the paper is devoted to proving Theorem \ref{thm:asympt12d Tc} and Theorem~\ref{thm:asympt12d Xi}.

\section{Proofs} \label{sec:proofs} 
The overall structure of our proofs is as follows:
The principal idea is to derive two different formulas for each of the two integrals 
\begin{equation}\label{mT}
  m_{\mu}^{(d)}(T) := \frac{1}{|\Sph^{d-1}|} \int_{|p| < \sqrt{2\mu}} \frac{1}{K_T(p)} \ud p
\end{equation}
and
\begin{equation} \label{mDelta}
  m_\mu^{(d)}(\Delta) := \frac{1}{|\Sph^{d-1}|} \int_{|p| < \sqrt{2\mu}} \frac{1}{E_\Delta(p)} \ud p.
\end{equation}
The first set of formulas is derived by studying the Birman-Schwinger operators 
\[
  B_T^{(d)} := \lambda V^{1/2} K_{T}^{-1} |V|^{1/2}\,
   \quad {\rm and} \quad  
  B_\Delta^{(d)} := \lambda V^{1/2} E_\Delta^{-1} |V|^{1/2}\,, 
\]
associated to the Schrödinger type operators $K_T + \lambda V$ and $E_\Delta + \lambda V$, respectively. 
In particular, spectral properties of these unbounded Schr\"odinger type operators naturally translate to their associated Birman-Schwinger operators, which are compact and as such much simpler to analyze. 
The second set of formulas is obtained by just calculating the integrals $m_{\mu}^{(d)}$ directly.

Indeed, for the critical temperature we obtain the following asymptotics, which, by combining them, immediately prove Theorem \ref{thm:asympt12d Tc}. 
\begin{prop}\label{lem.asym.Tc.d=2} 
Let $\mu > 0$. Let $V$ satisfy Assumption~\ref{assumption.W} and additionally $e_\mu^{(d)}< 0 $. Then, the critical temperature $T_c$ is positive and, as $\lambda \to 0$, we have that
\begin{align*}
m_\mu^{(d)}(T_c)  &= - \frac{\pi}{2 b_\mu^{(d)}(\lambda) } + o(1) \,, \\[1mm]
  m_{\mu}^{(d)}(T_c) &=  \mu^{d/2-1} \left( \ln \left(\frac{\mu}{T_c}\right) +\gamma+\ln\left(\frac{2c_d}{\pi}\right) +o(1)\right)\,.
\end{align*}
\end{prop}

\noindent
For the energy gap we obtain the following asymptotics, which, again by combining them, immediately prove Theorem \ref{thm:asympt12d Xi}. 
\begin{prop}\label{lem.asym.Delta.d=2}
Let $V$ satisfy Assumption \ref{assumption.d=2} and let $\mu > 0$. Then (by Proposition \ref{prop:alphaunique}) we have a strictly positive radially symmetric gap function $\Delta$ and associated energy gap $\Xi$, which, as $\lambda \to 0$, satisfy the asymptotics
\begin{equation*}
  \begin{aligned}
  \Xi & = \Delta(\sqrt{\mu})\big(1 + o(1)\big)
  \\
  m_\mu^{(d)}(\Delta) & = - \frac{\pi}{2 b_\mu^{(d)}(\lambda)}+ o(1)
  \\
  m_{\mu}^{(d)}(\Delta) & = \mu^{d/2-1} \left(\ln \left(\frac{\mu }{ \Delta(\sqrt{\mu})}\right) + \ln (2c_d) +o(1)\right)
  \end{aligned}
\end{equation*}
\end{prop}

\noindent
With a slight abuse of notation, using radiality of $\Delta$, we wrote $\Delta(\sqrt{\mu})$ instead of $\Delta(\sqrt{\mu}\hat{p})$ for some $\hat{p} \in \Sph^{d-1}$.

In the remainder of this paper, where we give the proofs of Propositions \ref{lem.asym.Tc.d=2} and \ref{lem.asym.Delta.d=2}, we shall frequently use the notation $\mathfrak{F}_\mu^{(d)}: L^1(\R^d)\to L^2(\Sph^{d-1})$ for the (scaled) Fourier transform restricted to the (rescaled) Fermi sphere,
\[
\left(  \mathfrak{F}_\mu^{(d)}\psi \right)(p) : = \frac{1}{(2\pi)^{d/2}} \int_{\R^d} \psi(x) e^{-i\sqrt{\mu}p\cdot x} \ud x\,. 
\]
Note that for an $L^1$-function, pointwise values of its Fourier transform are well-defined by the Riemann--Lebesgue lemma. 
(In particular the restriction to a co--dimension $1$ manifold of a sphere is well-defined.) 
\begin{rmk}
In \cite{cuenin.merz.2021}, Cuenin and Merz use the Tomas-Stein theorem to define 
$\mathfrak{F}_\mu^{(d)}$ on a larger space than $L^1(\R^d)$.	
With this they are able to prove a general version of \Cref{thm:asympt12d Tc} under slightly weaker conditions on $V$. However, we do not pursue this here, see \Cref{rmk.HS10}.
\end{rmk}

\subsection{Proof of Proposition \ref{lem.asym.Tc.d=2}}
\begin{proof}[Proof of \Cref{lem.asym.Tc.d=2}] The argument is divided into several steps. 
	\\[2mm]
	\textbf{1. A priori spectral information on $K_{T_c}+\lambda V$.} First note that, due to Theorem \ref{thm:lincrit} and Definition \ref{def:Tc}, the critical temperature $T_c$ is determined by the lowest eigenvalue of $K_T + \lambda V$ being $0$ exactly for $T=T_c$.
\\[2mm]
\textbf{2. Birman-Schwinger principle.}
Next, we employ the Birman-Schwinger principle, which says that the compact Birman-Schwinger operator
$B_T^{(d)} = \lambda V^{1/2} K_{T}^{-1} |V|^{1/2}$ (denoting $V(x)^{1/2} = \sgn(V(x)) |V(x)|^{1/2}$) has $-1$ as its lowest eigenvalue exactly for $T=T_c$, see \cite{frank.hainzl.naboko.seiringer,Hainzl.Seiringer.2008}.

Using the notation for the Fourier transform restricted to the rescaled Fermi sphere introduced above, we now decompose the Birman-Schwinger operator as
\[
  B_T^{(d)} = \lambda m_\mu^{(d)}(T) V^{1/2} (\mathfrak{F}_\mu^{(d)})^\dagger \mathfrak{F}_\mu^{(d)} |V|^{1/2} + \lambda V^{1/2} M_T^{(d)} |V|^{1/2},
\]
where $M_T^{(d)}$ is defined through the integral kernel
\begin{equation}\label{M_T_kernel}
	M_T^{(d)}(x,y) = \frac{1}{(2\pi)^{d}} \left[\int_{|p| < \sqrt{2\mu}} \frac{1}{K_T(p)} \left(e^{ip\cdot(x-y)} - e^{i\sqrt{\mu}p/|p| \cdot (x-y)}\right) \ud p
	+ \int_{|p| >\sqrt{2\mu}} \frac{1}{K_T} e^{ip\cdot(x-y)} \ud p\right].
\end{equation}
We claim that $V^{1/2} M_T^{(d)} |V|^{1/2}$ is uniformly bounded. 
\begin{lemma}\label{lem.MT.bdd.d=2}
Let $\mu > 0$. Let $V$ satisfy Assumption \ref{assumption.W}. Then we have for all $T\geq 0$
\[
  \norm{V^{1/2} M_T^{(d)} |V|^{1/2}}_{\textnormal{HS}} \leq C\,,
\]
where $C> 0$ denotes some positive constant and $\Vert \cdot \Vert_{\rm HS}$ is the Hilbert-Schmidt norm. 
\end{lemma}

\noindent
Armed with this bound, we have that for sufficiently small $\lambda$ that 
$1 + \lambda V^{1/2} M_T^{(d)} |V|^{1/2}$ is invertible, and hence
\[
  1 + B_T^{(d)} = (1 + \lambda V^{1/2} M_T^{(d)} |V|^{1/2})
    \left(1 + \frac{\lambda m_\mu^{(d)}(T)}{1 + \lambda V^{1/2} M_T^{(d)} |V|^{1/2}} V^{1/2} (\mathfrak{F}_\mu^{(d)})^\dagger \mathfrak{F}_\mu^{(d)} |V|^{1/2} \right)\,.
\]
Thus, the fact that $B_T^{(d)}$ has lowest eigenvalue $-1$ at $T=T_c$ is equivalent to 
\begin{equation} \label{eq:TcStep2}
  \lambda m_\mu^{(d)}(T)\mathfrak{F}_\mu^{(d)} |V|^{1/2} \frac{1}{1 + \lambda V^{1/2} M_T^{(d)} |V|^{1/2}} V^{1/2} (\mathfrak{F}_\mu^{(d)})^\dagger 
\end{equation}
having lowest eigenvalue $-1$, again at $T= T_c$, as it is isospectral to the rightmost operator on the right-hand-side above. 
(Recall that for bounded operators $A,B$, the operators $AB$ and $BA$ have the same spectrum apart from possibly at $0$. However, in our case, both operators are compact on an infinite dimensional space and hence $0$ is in both spectra.)

We now prove \Cref{lem.MT.bdd.d=2}.
\begin{proof}[Proof of \Cref{lem.MT.bdd.d=2}]
	We want to bound the integral kernel \eqref{M_T_kernel} of $M_T^{(d)}$ uniformly in $T$. Hence, we will bound $K_T \geq |p^2 - \mu|$. 
	The computation is slightly different in $d=1$ and $d=2$, so we do them separately.  
	\\[2mm]
	\underline{$d=1$.}
	The second integral in \eqref{M_T_kernel} is bounded by
	\[
	2 \int_{\vert p \vert>\sqrt{2\mu}} \frac{1}{\vert p^2-\mu \vert}\ud  p = \frac{2\arcoth \sqrt{2}}{\sqrt{\mu}}.
	\]
	For the first integral, we use that $\vert e^{ix}-e^{iy} \vert \leq \min\{\vert x- y \vert ,2\}$, $\vert p^2-\mu\vert \geq \sqrt{\mu}\vert \vert p \vert -\sqrt{\mu} \vert$, and increase the domain of integration to obtain the bound
	\begin{multline*}
		\frac{2}{\sqrt{\mu}}\int_{0}^{2\sqrt{\mu}} \frac{\min\left\{\vert \vert p-\sqrt{\mu}\vert \vert x-y\vert ,2\right\}}{\vert p -\sqrt{\mu} \vert}\ud p 
		= \frac{8}{ \sqrt{\mu}}\left[1+\ln\left(\max\left\{\frac{\vert x-y \vert \sqrt{\mu}}{2},1\right\}\right)\right] \\
		\leq \frac{8}{ \sqrt{\mu}}(1+\ln(1+\sqrt{\mu} \max\{\vert x\vert, \vert y\vert\}).
	\end{multline*}
	We conclude that $|M_T^{(1)}(x,y)|\lesssim \frac{1}{\sqrt{\mu}}( 1 +\ln(1+\sqrt{\mu} \max\{\vert x\vert, \vert y\vert\}))$.
	Hence, 
	\[
	\norm{V^{1/2}M_T^{(1)}|V|^{1/2}}_{\textnormal{HS}}^2 \lesssim \frac{1}{\mu}\left(\lVert V \rVert_{L^1(\R)}^2+ \lVert V\rVert_{L^1(\R)}\int_{\R} |V(x)|(1+\ln(1+\sqrt{\mu} \vert x\vert))^2 \ud x \right).
	\]
	\\[2mm]
	\underline{$d=2$}. 
	We first compute the angular integral. Note that 
	$\int_{\Sph^1} e^{ipx} \ud \omega(p) = 2\pi J_0(|x|)$, where $J_0$ is the zeroth order Bessel function.
	For the second integral in \eqref{M_T_kernel} we may bound $|p^2-\mu| \geq c p^2$. Up to some finite factor, the second integral is hence bounded by 
	\[
	\int_{\sqrt{2\mu}}^\infty \frac{1}{p} \vert J_0(p|x-y|) \vert \ud p \leq C\int_{\sqrt{2\mu}}^\infty \frac{1}{p^{1+\lambda}} |x-y|^{-\lambda} \ud p  \leq  C_\lambda |x-y|^{-\lambda},
	\]
	for any $0 < \lambda \leq 1/2$ since $|J_0(x)| \leq C$ and $\sqrt{x}J_0(x) \leq C$, see e.g.~\cite[(9.55f), (9.57a)]{bronstejn_taschenbuch_2012}.
	For the first integral we get the bound 
	\[  
	\int_0^{\sqrt{2\mu}} \frac{p}{|p^2 - \mu|} \abs{J_0(p|x-y|) - J_0(\sqrt{\mu}|x-y|)} \ud p.
	\]
	Here we use that $J_0$ is Lipschitz, since its derivative $J_{-1}$ is bounded (see e.g.~\cite[(9.55a),(9.55f)]{bronstejn_taschenbuch_2012}), so that  
	\[
	|J_0(x) - J_0(y)| \leq C |x-y|^{1/3} (|J_0(x)| + |J_0(y)|)^{2/3}
	\leq C |x-y|^{1/3} \left(x^{-1/3} + y^{-1/3}\right).
	\]
	That is 
	\[
	\left\vert J_0(p|x-y|) - J_0(\sqrt{\mu}|x-y|)\right\vert \leq C \frac{|p - \sqrt{\mu}|^{1/3}}{p^{1/3} + \sqrt{\mu}^{1/3}}.
	\]
	This shows that the first integral is bounded. We conclude that 
	$|M_T^{(2)}(x,y)|\lesssim 1 + \frac{1}{|x-y|^{\lambda}}$
	for any $0 < \lambda \leq 1/2$.
	Then, by the Hardy--Littlewood--Sobolev inequality \cite[Theorem 4.3]{analysis} we have that 
	\[
	\norm{V^{1/2}M_T^{(2)}|V|^{1/2}}_{\textnormal{HS}}^2 = \iint |V(x)| |M_T^{(2)}(x,y)| |V(y)| \ud x \ud y
	\lesssim \norm{V}_{L^1(\R^2)}^2 + \norm{V}_{L^p(\R^2)}^2
	\]
	for any $1 < p \leq 4/3$.
\end{proof}
\noindent
\textbf{3. First order.}
Evaluating \eqref{eq:TcStep2} at $T=T_c$ and expanding the geometric series to first order we get
\[
\begin{aligned}
  -1 & = \lambda m_\mu^{(d)}(T_c) 
    \inf\spec \left(\mathfrak{F}_\mu^{(d)} |V|^{1/2} \frac{1}{1 + \lambda V^{1/2} M_{T_c}^{(d)} |V|^{1/2}} V^{1/2} (\mathfrak{F}_\mu^{(d)})^\dagger\right)\\[1mm] 
    & = \lambda \, m_\mu^{(d)}(T_c) \,  \inf\spec \mathcal{V}_\mu^{(d)} (1 + O(\lambda)) = \lambda \, m_\mu^{(d)}(T_c) \, e_\mu^{(d)} (1 + O(\lambda))
\end{aligned}
\]
where we used $\mathcal{V}_\mu^{(d)} = \mathfrak{F}_\mu^{(d)} V (\mathfrak{F}_\mu^{(d)})^\dagger$. Since by assumption $e_\mu^{(d)} < 0$, this shows that $m_\mu^{(d)}(T_c) \to \infty$ as $\lambda \to 0$. 
\\[2mm]
\textbf{4. A priori bounds on $T_c$.} By \eqref{mT}, the divergence of $m_\mu^{(d)}$ as $\lambda \to 0$ in particular shows that $T_c/\mu \to 0$ in the limit $\lambda \to 0$.
\\[2mm]
\textbf{5. Calculation of the integral $m_\mu^{(d)}(T_c)$.} This step is very similar to \cite[Lemma 1]{Hainzl.Seiringer.2008} and \cite[Lemma 3.5]{Hainzl.Roos.ea.2022}, where the asymptotics have been computed for slightly different definitions of $m_\mu^{(d)}$ in three and one spatial dimension, respectively.
Integrating over the angular variable and substituting $s=\left\vert \frac{\vert p \vert^2}{\mu}-1\right\vert$, we get
\[
m_{\mu}^{(d)}(T_c) = \mu^{d/2-1}  \int_{0}^{1} \tanh\left(\frac{s}{2(T_c/\mu)}\right)\frac{(1+s)^{d/2-1}+(1-s)^{d/2-1}}{2s}  \ud s.
\]
According to \cite[Lemma 1]{Hainzl.Seiringer.2008}, 
\[ 
\lim_{T_c\downarrow0} \left(\int_0^1 \frac{\tanh\left(\frac{s}{2(T_c/\mu)}\right)}{s}\ud s-\ln \frac{\mu}{T_c}\right)=\gamma-\ln\frac{\pi}{2}.
\]
By monotone convergence, it follows that
\[
m_{\mu}^{(d)}(T_c) =  \mu^{d/2-1} 
	\left[\ln \frac{\mu}{T_c}+\gamma-\ln\frac{\pi}{2}+\int_{0}^{1} \frac{(1-s)^{d/2-1}+(1+s)^{d/2-1}-2}{2s} \ud s +o(1)\right]\,.
\]
The remaining integral equals $\ln c_d$ and we have thus proven the second item in Proposition~\ref{lem.asym.Tc.d=2}.

Combining this with the third step, one immediately sees that the critical temperature vanishes exponentially fast, $T_c \sim e^{1/\lambda e_\mu}$, as $\lambda \to 0$, recalling that $e_\mu^{(d)} < 0$ by assumption.
\\[2mm]
\textbf{6. Second order.}
Now, to show the universality, we need to compute the next order correction. 
To do so, we expand the geometric series in \eqref{eq:TcStep2} and employ first order perturbation theory, yielding that
\begin{equation} \label{eq:TcStep6}
  m_\mu^{(d)}(T_c)
= \frac{-1}{\lambda \longip{u}{\mathfrak{F}_\mu^{(d)} V (\mathfrak{F}_\mu^{(d)})^\dagger}{u} 
	- \lambda^2 \longip{u}{\mathfrak{F}_\mu^{(d)} V M_{T_c}^{(d)} V (\mathfrak{F}_\mu^{(d)})^\dagger}{u} + O(\lambda^3)},
\end{equation}
where $u$ is the (normalized) ground state (eigenstate of lowest eigenvalue) of $\mathfrak{F}_\mu^{(d)} V (\mathfrak{F}_\mu^{(d)})^\dagger$. 
(In case of a degenerate ground state, $u$ is the ground state minimizing the second order term.)

This second order term in the denominator of \eqref{eq:TcStep6} is close to $\mathcal{W}_\mu^{(d)}$. More precisely, it holds that
\begin{equation} \label{eq:Wconv}
  \lim_{\lambda\to 0}\longip{u}{\mathfrak{F}_\mu^{(d)} V M_{T_c}^{(d)} V (\mathfrak{F}_\mu^{(d)})^\dagger}{u} = \longip{u}{\mathcal{W}_\mu^{(d)}}{u}\,,
\end{equation}
which easily follows from dominated convergence, noting that $\frac{1}{K_T}$ increases to $\frac{1}{|p^2 - \mu|}$ as $T\to 0$.
We then conclude that 
\[
  \lim_{\lambda \to 0} \left(m_\mu^{(d)}(T_c) + \frac{\pi }{2 b_\mu^{(d)}(\lambda) }\right) = 0\,,
\]
since $\longip{u}{\lambda \mathcal{V}_\mu^{(d)} - \lambda^2 \mathcal{W}_\mu^{(d)}}{u} 
  = \inf\spec(\lambda \mathcal{V}_\mu^{(d)} - \lambda^2 \mathcal{W}_\mu^{(d)}) + O(\lambda^3)
  = \frac{\pi}{2} b_\mu^{(d)}(\lambda) + O(\lambda^3)$, 
again by first-order perturbation theory. 
This concludes the proof of Proposition~\ref{lem.asym.Tc.d=2}.
\end{proof}

\noindent
We conclude this subsection with several remarks, comparing our proof with those of similar results from the literature.

\begin{remark}[Structure here vs. in earlier papers on $T_c$] \label{rmk:Tccomp}
We compare the structure of our proof to that of the different limits in three dimensions
	\cite{Hainzl.Seiringer.2008, Henheik.2022, hainzl.seiringer.scat.length}:
\begin{itemize}
\item {\bf Weak coupling:}
The structure of the proof we gave here is quite similar to that of \cite{Hainzl.Seiringer.2008}, only they do Steps $5$ and $6$
in the opposite order. Also the leading term for $T_c$ was shown already in \cite{frank.hainzl.naboko.seiringer}, where a computation somewhat similar to Steps $1$--$4$ is given.
\item 	{\bf High denisty:} For $\mu \to \infty$, the structure of the proof  in \cite{Henheik.2022} is slightly different compared to the one given here. This is basically due to the facts that (i) the necessary a priori bound $T_c = o(\mu)$ already requires the Birman-Schwinger decomposition and (ii) the second order requires strengthened assumptions compared to the first order. To conclude, the order of steps in \cite{Henheik.2022} can be thought of as: $1$, $5$, $4$ (establishing $T_c = O(\mu)$), $2$, $3$, $4$ (establishing $T_c = o(\mu)$), $2$ (again), $6$. Here the final step is much more involved than in the other limits considered. 
\item {\bf Low density:}
As above, for the proof of the low density limit in \cite{hainzl.seiringer.scat.length} 
the structure is slightly different. 
One first needs the a priori bound $T_c = o(\mu)$ on the critical temperature
before one uses the Birman-Schwinger principle and decomposes the Birman-Schwinger operator.\footnote{Strictly speaking, in \cite{hainzl.seiringer.scat.length}, it is only proven that $T_c = O(\mu)$ (which is sufficient for applying the Birman-Schwinger principle), while the full $T_c = o(\mu)$ itself requires the Birman-Schwinger decomposition (see \cite[Remark~4.12]{ABLThesis} for details).} 
Also, the decomposition of the Birman-Schwinger operator is again different.
For the full decomposition and analysis of the Birman-Schwinger operator 
one needs also the first-order analysis, that is Step~$2$, which is done in two parts.
The order of the steps in \cite{hainzl.seiringer.scat.length} can then mostly be though of as:
$1$, $4$, $5$, $2$, $3$, $2$ (again), $6$.
\end{itemize}
\end{remark}

\subsection{Proof of \texorpdfstring{\Cref{lem.asym.Delta.d=2}}{Proposition \ref*{lem.asym.Delta.d=2}}}
\begin{proof}[Proof of \Cref{lem.asym.Delta.d=2}]
The structure of the proof is parallel to that of Proposition \ref{lem.asym.Tc.d=2} for the critical temperature.
	\\[2mm]
\textbf{1. A priori spectral information on $E_{\Delta}+\lambda V$.} First, it is proven in \cite[Lemma~2]{Hainzl.Seiringer.2008} that $\mathcal{F}_0$ has a unique minimizer $\alpha$ which has strictly positive Fourier transform. Using radiality of $V$, it immediately follows that this minimizer is rotationally symmetric (since otherwise rotating $\alpha$ would give a different minimizer) and hence also $\Delta = -2 \lambda \hat{V} \star \hat{\alpha}$ is rotation invariant. It directly follows from \cite[(43) and Lemma 3]{Hainzl.Seiringer.2008} that that $E_\Delta + \lambda V$ has lowest eigenvalue $0$, and that the minimizer $\alpha$ is the corresponding eigenfunction. 
\\[2mm]
\textbf{2. Birman-Schwinger principle.}
This implies, by means of the Birman-Schwinger principle, that the Birman-Schwinger operator $B_\Delta^{(d)} = \lambda V^{1/2} E_\Delta^{-1} |V|^{1/2}$ has $-1$ as its lowest eigenvalue. 
As in the proof of Proposition \ref{lem.asym.Tc.d=2}, we decompose it as
\[
  B_\Delta^{(d)} = \lambda m_\mu^{(d)}(\Delta) V^{1/2}(\mathfrak{F}_\mu^{(d)})^\dagger \mathfrak{F}_\mu^{(d)} |V|^{1/2}
    + \lambda V^{1/2} M_\Delta^{(d)} |V|^{1/2}
\]
and prove the second summand to be uniformly bounded.
\begin{lemma}\label{lem.MDelta.bdd.d=2}
	Let $\mu > 0$. Let $V$ satisfy Assumption \ref{assumption.W}.
Then, uniformly in  small $\lambda$, we have 
\[
  \norm{V^{1/2} M_\Delta^{(d)} |V|^{1/2}}_{\textnormal{HS}} \leq C\,. 
\]
\end{lemma}

\noindent
With this one may similarly factor 
\begin{equation}\label{eqn.decom.BDelta.d=2}
  1 + B_\Delta^{(d)} = (1 + \lambda V^{1/2} M_\Delta^{(d)} |V|^{1/2})
    \left(1 + \frac{\lambda m_\mu^{(d)}(\Delta)}{1 + \lambda V^{1/2} M_\Delta^{(d)} |V|^{1/2}} V^{1/2} (\mathfrak{F}_\mu^{(d)})^\dagger \mathfrak{F}_\mu^{(d)} |V|^{1/2} \right)
\end{equation}
and conclude that
\begin{equation}\label{eqn.TDelta}
  T_\Delta^{(d)} := \lambda m_\mu^{(d)}(\Delta) \mathfrak{F}_\mu^{(d)} |V|^{1/2} \frac{1}{1 + \lambda V^{1/2} M_\Delta^{(d)} |V|^{1/2}} V^{1/2} (\mathfrak{F}_\mu^{(d)})^\dagger
\end{equation}
has lowest eigenvalue $-1$. 
\begin{proof}[Proof of \Cref{lem.MDelta.bdd.d=2}]
	Note that $M_\Delta$ has kernel 
	\[
	M_\Delta(x,y) = \frac{1}{(2\pi)^{d}} \left[\int_{|p| < \sqrt{2\mu}} \frac{1}{E_\Delta(p)} \left(e^{ip\cdot(x-y)} - e^{i\sqrt{\mu}p/|p| \cdot (x-y)}\right) \ud p
	+ \int_{|p| >\sqrt{2\mu}} \frac{1}{E_\Delta(p)} e^{ip\cdot(x-y)} \ud p\right].
	\]
	We may bound this exactly as in the proof of \Cref{lem.MT.bdd.d=2} using that $E_\Delta(p) \geq |p^2 - \mu|$.
\end{proof}
\noindent \textbf{3. First order.}
Expanding the geometric series in \eqref{eqn.TDelta} to first order, we see that 
\[
\begin{aligned}
  -1 & = \lambda m_\mu^{(d)}(\Delta) 
    \inf\spec \left(\mathfrak{F}_\mu^{(d)} |V|^{1/2} \frac{1}{1 + \lambda V^{1/2} M_{\Delta}^{(d)} |V|^{1/2}} V^{1/2} (\mathfrak{F}_\mu^{(d)})^\dagger\right)
  \\ & = \lambda m_\mu^{(d)}(\Delta) \inf\spec \mathcal{V}_\mu^{(d)} (1 + O(\lambda))
  	= \lambda e_\mu^{(d)} m_\mu^{(d)}(\Delta)(1 + O(\lambda)).
\end{aligned}
\]
Hence, in particular, $m_\mu^{(d)}(\Delta) \sim -\frac{1}{\lambda e_\mu^{(d)}} \to \infty$ as $\lambda \to 0$. 
\\[2mm]
\textbf{4. A priori bounds on $\Delta$.}
We now prepare for the computation of the integral $m_\mu^{(d)}(\Delta)$ in terms of $\Delta(\sqrt{\mu})$. 
This requires two types of bounds on $\Delta$: One bound estimating the gap function $\Delta(p)$ at general momentum $p \in \R^d$ in terms of $\Delta(\sqrt{\mu})$ (see \eqref{eq:DeltaStep4a}), and one bound controlling the difference $|\Delta(p) - \Delta(q)|$ in some kind of Hölder-continuity estimate (see \eqref{eqn.Delta.holder}). 

\begin{lemma}\label{lem.Delta.formula}
Suppose that $V$ is as in \Cref{assumption.d=2}.  Then for $\lambda$ small enough 
\[ 
  \Delta(p) = f(\lambda) \left(\int_{\Sph^{d-1}} \hat V(p - \sqrt{\mu}q) \ud \omega (q) + \lambda \eta_\lambda(p) \right),
\]
where $f$ is some function of $\lambda$ and $\norm{\eta_\lambda}_{L^\infty(\R^d)}$ is bounded uniformly in $\lambda$.
\end{lemma}
\begin{proof}
Recall that $\alpha$ is the eigenfunction of $E_\Delta + \lambda V$ with lowest eigenvalue $0$.
Then, by the Birman-Schwinger principle, $\phi = V^{1/2}\alpha$ satisfies 
\[
  B_\Delta \phi = \lambda V^{1/2} \frac{1}{E_\Delta} |V|^{1/2} V^{1/2}\alpha = -\phi.
\]
With the decomposition \Cref{eqn.decom.BDelta.d=2} then $\phi$ is an eigenfunction of 
\[
\frac{\lambda m_\mu^{(d)}(\Delta)}{1 + \lambda V^{1/2} M_\Delta^{(d)} |V|^{1/2}} V^{1/2} (\mathfrak{F}_\mu^{(d)})^\dagger \mathfrak{F}_\mu^{(d)} |V|^{1/2}
\]
of eigenvalue $-1$. Thus, $\mathfrak{F}_\mu^{(d)} |V|^{1/2} \phi$ is an eigenfunction of $T_\Delta^{(d)}$ of (lowest) eigenvalue $-1$.
Now $u = |\Sph^{d-1}|^{-1/2}$ is the unique eigenfunction corresponding to the lowest eigenvalue of $\mathcal{V}_\mu^{(d)}$ by radiality of $V$ and the assumption $\hat{V} \le 0$ (see e.g.~\cite{frank.hainzl.naboko.seiringer}).
Hence, for $\lambda$ small enough, $u$ is the unique eigenfunction of $T_\Delta^{(d)}$ of smallest eigenvalue.
Thus, 
\[
  \phi = f(\lambda)\frac{1}{1 + \lambda V^{1/2} M_\Delta^{(d)} |V|^{1/2}} V^{1/2} (\mathfrak{F}_\mu^{(d)})^\dagger u 
    = f(\lambda) \left(V^{1/2} (\mathfrak{F}_\mu^{(d)})^\dagger u + \lambda \xi_\lambda\right)
\]
for some number $f(\lambda)$. The function $\xi_\lambda$ satisfies $\norm{\xi_\lambda}_{L^2(\R^d)} \leq C$ by \Cref{lem.MDelta.bdd.d=2}.
Noting that $\Delta = -2\widehat{|V|^{1/2}\phi}$ and bounding 
$\norm{\widehat{|V|^{1/2}\xi_\lambda}}_{L^\infty} \leq \norm{V}_{L^1}^{1/2} \norm{\xi_{\lambda}}_{L^2}$ 
we get the desired.
\end{proof}

\noindent
Evaluating the formula in \Cref{lem.Delta.formula} at $p = \sqrt{\mu}$ we get $|f(\lambda)|\leq C \Delta(\sqrt{\mu})$ for $\lambda$ small enough.
This in turn implies that
\begin{equation} \label{eq:DeltaStep4a}
\Delta(p) \leq C \Delta(\sqrt{\mu})\,. 
\end{equation}
For the Hölder-continuity, we have by rotation invariance
\begin{equation*}
\begin{aligned}
  & \abs{\int \hat V(p-\sqrt{\mu}r) - \hat V(q - \sqrt{\mu}r) \ud \omega(r)} =\abs{\int \hat V(\vert p\vert e_1-\sqrt{\mu}r) - \hat V(\vert q\vert e_1 - \sqrt{\mu}r) \ud \omega(r)} 
  \\ & \quad 
    = \abs{\frac{1}{(2\pi)^{d/2}} \int_{\R^d} \ud x \left(V(x) \left(e^{i\vert p\vert x_1} - e^{i\vert q\vert x_1}\right) \int_{\Sph^{d-1}} e^{-i\sqrt{\mu}x\cdot r} \ud \omega(r)\right)}
  \\ & \quad 
    \leq C_\eps \mu^{-\eps/2}|\vert p\vert -\vert q\vert|^{\eps} \int \ud x \left(|V(x)| (\sqrt{\mu}|x|)^{\eps} \abs{\int_{\Sph^{d-1}} e^{-i\sqrt{\mu}x\cdot r} \ud \omega(r)}\right),
\end{aligned}
\end{equation*}
for any $0<\eps\leq 1$.
For $d=2$ we have $V\in L^1(\R^2)$ and 
\[
	\abs{\int_{\Sph^{d-1}} e^{-i\sqrt{\mu}xr} \ud \omega(r)} = |J_0(\sqrt{\mu}|x|)| \leq (\sqrt{\mu}|x|)^{-1/2}.
\]
For $d=1$ we have $|x|^\eps V\in L^1(\R)$ for some $\eps > 0$ and 
\[
	\abs{\int_{\Sph^{d-1}} e^{-i\sqrt{\mu}x\cdot r} \ud \omega(r)} = 2|\cos(\sqrt{\mu}|x|)| \leq 2.
\] 
We conclude that with $\eps = 1/2$ for $d=2$ and small enough $\eps > 0$ for $d=1$
\begin{equation}\label{eqn.Delta.holder}
	|\Delta(p) - \Delta(q)| 
	\leq C |f(\lambda)| \left( \mu^{-\eps/2}|\vert p\vert -\vert q\vert|^{\eps} + \lambda\right) 
	\leq C |\Delta(\sqrt{\mu})| \left( \mu^{-\eps/2}|\vert p\vert -\vert q\vert |^{\eps} + \lambda\right).
\end{equation}
Additionally, since $m_\mu^{(d)}(\Delta)\to \infty$ we have that $\Delta(p)\to 0$ at least for some $p \in \R^d$ by \eqref{mDelta}. 
Then it follows from Lemma~\ref{lem.Delta.formula} that $f(\lambda)\to 0$, 
i.e. that $\Delta(p) \to 0$ for all $p$.
\\[2mm]
\textbf{5. Calculation of the integral $m_\mu^{(d)}(\Delta)$.}
Armed with the apriori bounds \eqref{eq:DeltaStep4a} and \eqref{eqn.Delta.holder}, we can now compute the integral $m_\mu^{(d)}(\Delta)$.
Carrying out the angular integration and substituting $s = \abs{\frac{|p|^2-\mu}{\mu}}$ we have 
\[
\begin{aligned}
  m_\mu^{(d)}(\Delta)
    & = \frac{\mu^{d/2-1}}{2} \left[
      \int_0^1 \left(\frac{(1-s)^{d/2-1}-1}{\sqrt{s^2 + x_-(s)^2}} + \frac{(1+s)^{d/2-1}-1}{\sqrt{s^2 + x_+(s)^2}}\right) \ud s
      \right.
    \\ & \qquad 
      \left.
      + \int_0^1 \left(\frac{1}{\sqrt{s^2 + x_-(s)^2}} + \frac{1}{\sqrt{s^2 + x_+(s)^2}} \right)\ud s
      \right]\,,
\end{aligned}
\]
where $x_{\pm}(s) = \frac{\Delta(\sqrt{\mu}\sqrt{1\pm s})}{\mu}$. 
By dominated convergence, using that $x_\pm(s) \to 0$, 
the first integral is easily seen to converge to 
\[
  \int_0^1 \left(\frac{(1-s)^{d/2-1}-1}{s} + \frac{(1+s)^{d/2-1}-1}{s}\right) \ud s = 2 \ln c_d\,
\]
for $\lambda \to 0$.
For the second integral, we will now show that
\[
  \int_0^1 \left(\frac{1}{\sqrt{s^2 + x_\pm(s)^2}} - \frac{1}{\sqrt{s^2 + x_\pm(0)^2}} \right) \ud s \to 0\,. 
\]
In fact, the integrand is bounded by
\[
\begin{aligned}
& \abs{\frac{1}{\sqrt{s^2 + x_\pm(s)^2}} - \frac{1}{\sqrt{s^2 + x_\pm(0)^2}}}
	\\ & \quad 
		= \frac{\abs{x_{\pm}(0)^2 - x_{\pm}(s)^2}}{\sqrt{s^2 + x_\pm(s)^2}\sqrt{s^2 + x_\pm(0)^2}(\sqrt{s^2 + x_\pm(s)^2} + \sqrt{s^2 + x_\pm(0)^2})}
	\\ & \quad 
		\leq \frac{C x_{\pm}(0) (s^\eps + \lambda)}{\sqrt{s^2 + x_\pm(s)^2}\sqrt{s^2 + x_\pm(0)^2}},
\end{aligned}	
\]
using the Hölder continuity from \eqref{eqn.Delta.holder}.
By continuity of $\hat V$ there exists some $s_0$ (independent of $\lambda$) such that for $s < s_0$ we have $x_\pm(s) \geq c x_\pm(0)$. 
We now split the integration into $\int_0^{s_0}$ and $\int_{s_0}^1$.
For the first we have
\[
\begin{aligned}
	\int_0^{s_0} \abs{\frac{1}{\sqrt{s^2 + x_\pm(s)^2}} - \frac{1}{\sqrt{s^2 + x_\pm(0)^2}}} \ud s 
	& \leq C \int_0^{s_0} \frac{x_\pm(0)}{s^2 + x_\pm(0)^2} (s^\eps + \lambda)\ud s  
	= O(x_\pm(0)^\eps + \lambda)\,. 
\end{aligned}
\]
For the second we have 
\[
\begin{aligned}
	\int_{s_0}^1 \abs{\frac{1}{\sqrt{s^2 + x_\pm(s)^2}} - \frac{1}{\sqrt{s^2 + x_\pm(0)^2}}} \ud s 
	& \leq C \int_{s_0}^1 x_\pm(0) \frac{s^\eps + \lambda}{s^2} \ud s = O(x_\pm(0))\,. 
\end{aligned}
\]
Collecting all the estimates, we have thus shown that $  m_\mu^{(d)}(\Delta)$ equals
\begin{align*}
  & \mu^{d/2-1}
      \left( \ln c_d+
      \int_0^1 \frac{1}{\sqrt{s^2 + \Delta(\sqrt{\mu})^2/\mu^2}} \ud s
      +o(1)\right) \\
=\, &\mu^{d/2-1} \left( \ln c_d+\ln \left(\frac{\mu+\sqrt{\mu^2+\Delta(\sqrt{\mu})^2}}{\vert \Delta(\sqrt{\mu})\vert}\right)+o(1)\right)
=\mu^{d/2-1} \ln \left(\frac{2\mu c_d}{\vert \Delta(\sqrt{\mu})\vert}+o(1)\right)\,. 
\end{align*}
This proves the third inequality in \Cref{lem.asym.Delta.d=2}.

Combining this with the third step, one immediately sees that the gap function evaluated on the Fermi sphere vanishes exponentially fast, $\Delta(\sqrt{\mu}) \sim e^{1/\lambda e_\mu}$, as $\lambda \to 0$, recalling that $e_\mu^{(d)} < 0$ by assumption.
\\[2mm]
\textbf{6. Second order.}
To obtain the next order, we recall that $T_\Delta^{(d)}$ has lowest eigenvalue $-1$ (see \eqref{eqn.TDelta}), and hence, by first-order perturbation theory,
\begin{equation}\label{eqn.mDelta.pert.exp}
  m_\mu^{(d)}(\Delta)
  = \frac{-1}{\lambda \longip{u}{\mathfrak{F}_\mu^{(d)} V (\mathfrak{F}_\mu^{(d)})^\dagger}{u} 
    - \lambda^2 \longip{u}{\mathfrak{F}_\mu^{(d)} V M_{\Delta}^{(d)} V (\mathfrak{F}_\mu^{(d)})^\dagger}{u} + O(\lambda^3)},
\end{equation}
where $u(p) = |\Sph^{d-1}|^{-1/2}$ is the constant function on the sphere.
Recall that $u$ is the unique ground state of $ \mathcal{V}_\mu^{(d)}$. 

In the second order term we have that
\[
  \lim_{\lambda \to 0} \longip{u}{\mathfrak{F}_\mu^{(d)} V M_{\Delta}^{(d)} V (\mathfrak{F}_\mu^{(d)})^\dagger}{u} = \longip{u}{\mathcal{W}_\mu^{(d)}}{u}\,,
\]
which follows from a simple dominated convergence argument as for $T_c$, noting that $\Delta(p)\to 0$ pointwise. 

By again employing first--order perturbation theory, similarly to the last step in the proof of Proposition \ref{lem.asym.Tc.d=2}, we conclude the second equality in \Cref{lem.asym.Delta.d=2}.
\\[2mm]
\textbf{7. Comparing $\Delta(\sqrt{\mu})$ to $\Xi$.}
To prove the first equality in \Cref{lem.asym.Delta.d=2} we separately prove upper and lower bounds. 
The upper bound is immediate from 
\[
\Xi = \inf_{p \in \R^d} E_\Delta(p) = \inf_{p \in \R^d} \sqrt{|p^2 - \mu| + \Delta(p)^2} \leq \Delta(\sqrt{\mu})\,.
\]
Hence, for the lower bound, take $p \in \R^d$ with $\sqrt{|p^2 - \mu|}\leq \Xi \leq \Delta(\sqrt{\mu})$.
Then by \eqref{eqn.Delta.holder}
\[
  \Delta(p) \geq \Delta(\sqrt{\mu}) - |\Delta(p) - \Delta(\sqrt{\mu})| \geq \Delta(\sqrt{\mu}) - C \Delta(\sqrt{\mu}) \left(|\vert p\vert - \sqrt{\mu}|^\eps + \lambda\right)
  \geq \Delta(\sqrt{\mu})(1 + o(1)).
\]
In combination with the upper bound, we have thus shown that $\Xi = \Delta(\sqrt{\mu})(1 + o(1))$ as desired. This 
concludes the proof of Proposition \ref{lem.asym.Delta.d=2}. 
\end{proof}

\noindent
We conclude this subsection with several remarks, comparing our proof with those of similar results from the literature.

\begin{remark}[Structure here vs. in earlier papers on $\Xi$] \label{rmk:Xicomp}
We now compare the proof above to the proofs of the three different limits in $3$ dimensions
\cite{Hainzl.Seiringer.2008,lauritsen.energy.gap.2021,Henheik.Lauritsen.2022}: 
\begin{itemize}
\item {\bf Weak coupling:}
The structure of our proof here is very similar to that of \cite{Hainzl.Seiringer.2008}.  
Essentially, only the technical details in \Cref{lem.MDelta.bdd.d=2} and the calculation of $m_\mu^{(d)}(\Delta)$ in Step 5 are different.
\item {\bf High density:}
For the high-density limit in \cite{Henheik.Lauritsen.2022}, we needed some additional a priori bounds on $\Delta$ 
before we could employ the Birman-Schwinger argument. Apart from that, in \cite{Henheik.Lauritsen.2022} 
the comparison of $\Delta(\sqrt{\mu})$ and $\Xi$ are done right after these a priori bounds. 
Additionally, since one starts with finding a priori bounds on $\Delta$, 
one does not need the first-order analysis in Step 3.
One may think of the structure in \cite{Henheik.Lauritsen.2022} as being ordered in the above steps as follows:
$4$, $7$, $1$, $2$, $4$ (again), $5$, $6$. 
\item {\bf Low density:}
For the low-density limit in \cite{lauritsen.energy.gap.2021} the structure is quite different.
Again, one first needs some a priori bounds on $\Delta$ before one can use the Birman-Schwinger argument.
One then improves these bounds on $\Delta$ using the Birman-Schwinger argument, 
which in turn can be used to get better bounds on the error term in the decomposition of the Birman--Schwinger operator. In this sense, the Steps $2$--$4$ are too interwoven to be meaningfully separated. Also, Step $5$ is done in two parts. 
\end{itemize}
\end{remark}

\subsection{Proof of Proposition \ref{prop:Wmu}}\label{sec:pf_Wmu}
Note that $\mathcal{W}_\mu^{(d)}=\mathfrak{F}_\mu^{(d)} V M_0^{(d)} V (\mathfrak{F}_\mu^{(d)})^\dagger$,  where $M_0^{(d)}$ is defined in \eqref{M_T_kernel}.
By Lemma~\ref{lem.MT.bdd.d=2}, $V^{1/2} M_0^{(d)} V^{1/2}$ is Hilbert-Schmidt. 
The integral kernel of $\mathcal{W}_\mu^{(d)}$ is bounded by 
\begin{equation}
\vert \mathcal{W}_\mu^{(d)}(p,q)\vert \leq \frac{1}{(2\pi)^{d}}\int_{\R^{2d}} \vert V(x) \vert  \vert M_0^{(d)}(x,y)\vert  \vert V(y)\vert \ud x \ud y \leq \frac{1}{(2\pi)^{d}}\lVert V\rVert_1\lVert V^{1/2} M_0^{(d)} V^{1/2}\rVert_{\rm HS}.
\end{equation}
It follows that 
$\lVert \mathcal{W}_\mu^{(d)} \rVert_{\rm HS} 
\leq \frac{\abs{\mathbb{S}^{d-1}}}{(2\pi)^{d}}\lVert V \rVert_1 \lVert V^{1/2} M_0^{(d)} V^{1/2}\rVert_{\rm HS}$.
\qed
\\[7mm]
\emph{Acknowledgments.} 
We thank Robert Seiringer for comments on the manuscript.
J.H. gratefully acknowledges partial financial support by the ERC
Advanced Grant `RMTBeyond' No.~101020331.

\printbibliography

\end{document}